\RequirePackage{fix-cm}
\documentclass[smallcondensed]{svjour3}     % onecolumn (ditto)
\smartqed  % flush right qed marks, e.g. at end of proof
\usepackage{graphicx,url}
\usepackage{amssymb,amsmath}
\usepackage{color,ulem}

\newcommand{\remove}[1]{}
\newcommand{\vbar}{\:\bigg{|}\:}
\newcommand{\Vbar}{\:\Bigg{|}\:}

\begin{document}

\title{Connections Between Construction D and Related Constructions of Lattices
\thanks{The work of W. Kositwattanarerk was conducted in part while the author was at the Division of Mathematical Sciences, School of Physical and Mathematical Sciences, Nanyang Technological University, Singapore. The research of W. Kositwattanarerk and F. Oggier for this work is supported by the Singapore National Research Foundation under Research Grant NRF-RF2009-07. The material in this paper was presented in part at the International Workshop on Coding and Cryptography, Bergen, Norway, April 2013.}
}


\author{Wittawat Kositwattanarerk         \and
        Fr\'{e}d\'{e}rique Oggier %etc.
}


\institute{Wittawat Kositwattanarerk \at
              Department of Mathematics \\
              Faculty of Science \\
              Mahidol University, Thailand \\
              \email{wittawat.kos@mahidol.ac.th}           %  \\
           \and Fr\'{e}d\'{e}rique Oggier \at
              Division of Mathematical Sciences \\
              School of Physical and Mathematical Sciences \\
              Nanyang Technological University, Singapore \\
              \email{frederique@ntu.edu.sg} 
}

\date{Received: date / Accepted: date}

\maketitle

\begin{abstract}
Most practical constructions of lattice codes with high coding gains are multilevel constructions where each level corresponds to an underlying code component. Construction D, Construction D$'$, and Forney's code formula are classical constructions that produce such lattices explicitly from a family of nested binary linear codes. In this paper, we investigate these three closely related constructions along with the recently developed Construction A$'$ of lattices from codes over the polynomial ring $\mathbb{F}_2[u]/u^a$. We show that Construction by Code Formula produces a lattice packing if and only if the nested codes being used are closed under Schur product, thus proving the similarity of Construction D and Construction by Code Formula when applied to Reed-Muller codes. In addition, we relate Construction by Code Formula to Construction A$'$ by finding a correspondence between nested binary codes and codes over $\mathbb{F}_2[u]/u^a$. This proves that any lattice constructible using Construction by Code Formula is also constructible using Construction A$'$. Finally, we show that Construction A$'$ produces a lattice if and only if the corresponding code over $\mathbb{F}_2[u]/u^a$ is closed under shifted Schur product.

\remove{We investigate four closely related constructions of lattices from linear codes: the classical Construction D, Construction D$'$, Construction by Code Formula, and the recently developed Construction A$'$. These constructions have been proven useful and result in efficient encoding and decoding algorithms for Barnes-Wall lattices. Here, we analyze their applications in a general setting. \\
We give an exact condition for codes over polynomial rings to produce a lattice under Construction A$'$. \\
Prompted by these applications, we analyze the three constructions in the general settings. \\
Construction $\overline{\mbox{D}}$ is a generalization that we propose of the construction given in \cite{F}, and Construction A$'$ is introduced in \cite{HVB_ISIT,HVB} as an extension of Construction A to codes over polynomial rings. We prove that Construction $\overline{\mbox{D}}$ and Construction A$'$ do not always produce a lattice, and provide a necessary and sufficient condition for Construction $\overline{\mbox{D}}$ to give a lattice. We also show that any lattices constructible using Construction $\overline{\mbox{D}}$ are also constructible using Construction A$'$.}

\keywords{Lattices \and lattices from codes \and coset codes \and Barnes-Wall lattices \and Schur product of codes}
\end{abstract}

%******************************************************************************%
%
%
%
%******************************************************************************%
\section{Introduction}

Connections between lattices and linear codes are classically studied (see e.g. \cite{CS}). Lattices constructed from codes often inherit certain properties from the underlying codes and have manageable encoding and decoding complexity. In particular, Construction D and D$'$ of Barnes and Sloane \cite{BS,CS} produce lattice packings from a family of nested binary linear codes where artificial ``levels'' are created by using an increasing power of 2. These constructions are well-known for the construction of Barnes-Wall lattices from Reed-Muller codes. \\

In the classical definition of Construction D and D$'$ \cite{BS,CS}, there are restrictions on the minimum distance of the codes being used. Nonetheless, new classes of codes and recent applications of lattice codes call for a reinvestigation of these conditions. In particular, Construction D$'$ is used in conjunction with low-density parity-check (LDPC) codes to produce what is called LDPC lattices in \cite{SBP}. Construction D is exploited in \cite{SSP}, making possible a construction of lattices from turbo codes. In both applications the restrictions on the minimum distance of the codes are lifted and relaxed constructions are studied. Here, we will also consider the constructions without the minimum distance conditions so that our discussions apply in the general case. \\

Construction by Code Formula is a reformulation of Forney's code formula \cite{F1,F2} where the lattice in consideration is decomposed as
\[\Lambda=C_0+2C_1+\ldots+2^{a-1}C_{a-1}+2^a\mathbb{Z}^n\]
where $a$ is the 2-depth of the lattice and $C_0,\ldots,C_{a-1}$ are binary codes. This decomposition exists for many notable lattices, including $E_8$, Barnes-Wall lattices, and the Leech lattices. While Forney initially introduces the code formula as a special case where a lattice can be decomposed completely into a chain of coset codes, many recent applications take interest in using this code formula as a construction and representation of lattices \cite{HVB_ISIT,HVB,OSB,YLW}. Often, the code formula as a construction of lattices is also called Construction D. To avoid confusion, we use the term Construction D to refer to the original construction given in \cite{BS,CS} and refer to the construction based on code formula as Construction by Code Formula. Indeed, we will show that the two constructions coincide under a certain condition. \\

Construction $A'$ is an extension of Construction A recently proposed by Harshan, Viterbo, and Belfiore \cite{HVB_ISIT,HVB}. This construction combines codes of different levels in a code formula and generates lattices from a single code over a polynomial ring $\mathbb{F}_2[u]/u^a$. It was shown in \cite{HVB_ISIT} that an encoding of Barnes-Wall lattices using Construction A$'$ is equivalent to the encoding using the traditional code formula with Reed-Muller codes. \\

{Prompted by this new wave of interest, in this paper we analyze and find connections between the three constructions of lattices from codes: Construction D, Construction by Code Formula, and Construction A$'$.} In Section 2, we give the definition for Construction D, Construction D$'$, Construction by Code Formula, and an example distinguishing the three constructions and demonstrating that Construction by Code Formula does not always produce a lattice. Section 3 provides a necessary and sufficient condition for Construction by Code Formula to output a lattice and relates this construction to Construction D. The definition for Construction A$'$ is given is Section 4. Here, we discuss the relationships between Construction by Code Formula and Construction A$'$, along with a necessary and sufficient condition for Construction A$'$ to produce a lattice. We conclude the paper in Section 5.

\section{Definitions and Example}

Let $\psi$ be the natural embedding of $\mathbb{F}_2^n$ into $\mathbb{Z}^n$, where $\mathbb{F}_2$ is the binary field. We first recall the definition of Construction D and D$'$ \cite{CS}. Note that we scale Construction D as suggested in \cite[page 236]{CS} so that the two constructions are comparable.

\begin{definition}[Construction D (scaled)]\label{DefinitionD}
Let $C_0\subseteq{}C_1\subseteq\ldots\subseteq{}C_{a-1}\subseteq{}C_a=\mathbb{F}_2^n$ be a family of nested binary linear codes where the minimum distance of $C_i$ is at least $4^{a-i}/\gamma$ where $\gamma=$ 1 or 2. Let $k_i=\dim(C_i)$ and let $\mathbf{b}_1,\mathbf{b}_2,\ldots,\mathbf{b}_n$ be a basis of $\mathbb{F}_2^n$ such that $\mathbf{b}_1,\ldots,\mathbf{b}_{k_i}$ span $C_i$. The lattice $\Lambda_D$ consists of all vectors of the form
\[\sum_{i=0}^{a-1}{{2^i}\sum_{j=1}^{k_i}{\alpha_j^{(i)}\psi(\mathbf{b}_j)}}+2^a\mathbf{l}\]
where $\alpha_j^{(i)}\in\{0,1\}$ and $\mathbf{l}\in\mathbb{Z}^n$.
\end{definition}

\begin{definition}[Construction D$'$]
Let $C_0\subseteq{}C_1\subseteq\ldots\subseteq{}C_{a-1}$ be a family of nested binary linear codes where the minimum distance of $C_i$ is at least $\gamma\cdot4^{a-i-1}$ where $\gamma=$ 1 or 2. Let $r_i=n-\dim(C_i)$ and {$r_a=0$}. Let $\mathbf{h}_1,\mathbf{h}_2,\ldots,\mathbf{h}_n$ be a basis for $\mathbb{F}_2^n$ such that $C_i$ is defined by the parity-check vectors $\mathbf{h}_1,\ldots,\mathbf{h}_{r_i}$. Let $\Lambda_{D'}$ be the lattice consisting of all vectors $\mathbf{x}\in\mathbb{Z}^n$ satisfying the congruences
\[\mathbf{x}\cdot\psi(\mathbf{h}_j)\equiv0\pmod{2^{i+1}}\]
for all $i\in\{0,\ldots,a-1\}$ and {$r_{i+1}\leq{}j\leq{}r_{i}$}.
\end{definition}

Although the classical definition of Construction D (resp. Construction D$'$) requires that the minimum distance of $C_i$ is at least $4^{a-i}/\gamma$ (resp. $\gamma\cdot4^{a-i-1}$), it is possible to apply the constructions even when the conditions are not met \cite{SBP,SSP}. In this paper, we will relax these minimum distance conditions since they do not affect our discussions on constructions of lattices. Fundamental parameters of lattices from these relaxed constructions can be found in \cite{SBP,SSP}. The next definition gives a multilevel construction of lattices based on the so-called code formula \cite{F1}.

\begin{definition}[Construction by Code Formula\footnote{This construction was earlier used by the name Construction $\overline{\mbox{D}}$ in \cite{KO}.}]
Let $C_0\subseteq{}C_1\subseteq\ldots\subseteq{}C_{a-1}\subseteq{}C_a=\mathbb{F}_2^n$ be a family of nested binary linear codes. Let
\[\Gamma_{CF}=\psi(C_0)+2\psi(C_1)+\ldots+2^{a-1}\psi(C_{a-1})+2^a\mathbb{Z}^n.\]
\end{definition}

{When $a=1$ (i.e., $C_0$ is the only code in consideration), Construction D, Construction D$'$, and Construction by Code Formula coincide and reduce to what is called Construction A \cite{CS}. Hence, Construction A produces lattices from a single binary code, and one may view the three constructions presented here as different generalizations of Construction A. We refer the reader to \cite{CS} for other generalizations, such as Construction A for codes over $\mathbb{Z}_{q}$.}

Forney states that a lattice is not necessarily expressible using this code formula, but most lattices that are useful in practice are \cite{F1}. We will see in Example \ref{Example} in addition that the set $\Gamma_{CF}$ itself may not be a lattice, so we denote by $\Lambda_{CF}$ the smallest lattice that contains $\Gamma_{CF}$. We first state a trivial observation concerning $\Lambda_{CF}$.

\begin{lemma}\label{LemmaBasic}
Let $C_0\subseteq{}C_1\subseteq\ldots\subseteq{}C_{a-1}\subseteq{}C_a=\mathbb{F}_2^n$ be a family of nested binary linear codes. Then, $\Lambda_{D}\subseteq\Lambda_{CF}$.
\end{lemma}

\begin{proof}
Let $\mathbf{b}_1,\mathbf{b}_2,\ldots,\mathbf{b}_n$ be a basis of $\mathbb{F}_2^n$ as in Definition \ref{DefinitionD}. Since $\Lambda_{CF}$ is closed under addition and $2^i\psi(\mathbf{b}_j)\in2^i\psi(C_i)$ for each $i$ and $1\leq{}j\leq{}k_i$, we have $\Lambda_{D}\subseteq\Lambda_{CF}$.
\end{proof}

\remove{Since the volume of $\Lambda_D$ is $2^{\sum_{i=1}^a{n-k_i}}$, it follows that the volume of $\Lambda_{CF}$ is at most $2^{\sum_{i=1}^a{n-k_i}}$, with equality when $\Lambda_{D}=\Lambda_{CF}$.}
One may observe that Construction D (resp. Construction D$'$) depends on the choice of generators (resp. parity conditions) of the codes whereas Construction by Code Formula is independent of these choices. We will see in Example \ref{Example} that the three constructions produce distinct lattices and $\Gamma_{CF}$ from Construction by Code Formula may not be a lattice.

\begin{example}\label{Example}
Consider the nested binary linear codes
\[C_0\subseteq{}C_1\subseteq{}C_2\]
where
\begin{align*}
C_0 & =\langle(1,1,0,0),(1,0,1,0)\rangle, \\
C_1 & =\langle(1,1,0,0),(1,0,1,0),(1,0,0,1)\rangle,\mbox{ and} \\
C_2 & =\mathbb{F}_2^4.
\end{align*}
Alternatively, let $\{\mathbf{b}_1,\mathbf{b}_2,\mathbf{b}_3,\mathbf{b}_4\}$ and $\{\mathbf{h}_1,\mathbf{h}_2,\mathbf{h}_3,\mathbf{h}_4\}$ be a basis for $\mathbb{F}_2^4$ where
\[\begin{array}{l}
\left.
\begin{array}{l}
\left.
\begin{array}{l}
\mathbf{b}_1=(1,1,0,0) \\
\mathbf{b}_2=(1,0,1,0)
\end{array}
\right\}\mbox{span }C_0 \\
\:\mathbf{b}_3=(1,0,0,1)
\end{array}
\right\}\mbox{span }C_1 \\
\:\:\mathbf{b}_4=(1,0,0,0)
\end{array}
\]
and
\[\begin{array}{l}
\left.
\begin{array}{l}
\left.
\begin{array}{l}
\mathbf{h}_1=(1,1,1,1) \\
\end{array}
\right\}\mbox{check }C_1 \\
\:\mathbf{h}_2=(0,0,0,1)
\end{array}
\right\}\mbox{check }C_0 \\
\:\:\mathbf{h}_3=(1,0,0,0) \\
\:\:\mathbf{h}_4=(0,1,0,0).
\end{array}\]
Then,
\[\Lambda_D=\left\{
\begin{array}{l}
(\alpha_1^{(0)}(1,1,0,0)+\alpha_2^{(0)}(1,0,1,0)) \\
+2(\alpha_1^{(1)}(1,1,0,0)+\alpha_2^{(1)}(1,0,1,0)+\alpha_3^{(1)}(1,0,0,1)) \\
+4\mathbf{l}
\end{array}
\Vbar
\begin{array}{l}
\alpha_j^{(i)}\in\{0,1\}, \\
\mathbf{l}\in\mathbb{Z}^4
\end{array}
\right\}\]
and
\[{\Lambda_{D'}=\left\{\mathbf{x}\in\mathbb{Z}^4\vbar
\begin{array}{l}
(0,0,0,1)\cdot\mathbf{x}\equiv0\pmod2,\mbox{ and} \\
(1,1,1,1)\cdot\mathbf{x}\equiv0\pmod4
\end{array}
\right\}}.\]
Also,
\begin{align*}
\Gamma_{CF} & =\psi(C_0)+2\psi(C_{1})+4\mathbb{Z}^4 \\
& = \{\mathbf{c}_0+2\mathbf{c}_1+4\mathbf{l}\mid\mathbf{c}_0\in\psi(C_0),\mathbf{c}_1\in\psi(C_1),\mbox{and }\mathbf{l}\in\mathbb{Z}^4\}.
\end{align*}

{We will show next that \textit{i)} $\Gamma_{CF}\subsetneqq\Lambda_{CF}$, \textit{ii)} $\Lambda_D\subsetneqq\Lambda_{CF}$, \textit{iii)} $\Lambda_D\neq\Lambda_{D'}$, and \textit{iv)} $\Lambda_{D'}\subsetneqq\Lambda_{CF}$.}
\begin{description}
\item[{\textit{i)} $\Gamma_{CF}\subsetneqq\Lambda_{CF}$:}] Since every entry of the vectors in $\psi(C_0)+2\psi(C_{1})$ is at most 3, every element of $\Gamma_{CF}$ reduces mod 4 to an element in $\psi(C_0)+2\psi(C_{1})$. Now, since $(1,1,0,0)$, $(1,0,1,0)$, $(0,1,1,0)\in\Gamma_{CF}$ but $(2,0,0,0)=(1,1,0,0)+(1,0,1,0)-(0,1,1,0)$ does not reduce mod 4 to an element in $\psi(C_0)+2\psi(C_{1})$, we can conclude that $(2,0,0,0)\notin\Gamma_{CF}$ and $\Gamma_{CF}$ is not a lattice. Therefore, $\Gamma_{CF}\subsetneqq\Lambda_{CF}$.
\item[{\textit{ii)} $\Lambda_D\subsetneqq\Lambda_{CF}$:}] It follows from Lemma \ref{LemmaBasic} that $\Lambda_D\subseteq\Lambda_{CF}$. Now, every element $(a_1,a_2,a_3,a_4)$ of $\Lambda_D$ must satisfy
\begin{equation}\label{EquationD}
a_1\equiv{}a_2+a_3+a_4\pmod4.
\end{equation}
However, since $(0,1,1,0)\in\Lambda_{CF}$ does not satisfy (\ref{EquationD}), we conclude that $\Lambda_D\subsetneqq\Lambda_{CF}$.
\item[{\textit{iii)} $\Lambda_D\neq\Lambda_{D'}$:}] The vector $(1,3,0,0)\in\Lambda_{D'}$ also does not satisfy (\ref{EquationD}), so $\Lambda_{D'}\not\subset\Lambda_D$. Furthermore, $(1,1,0,0)\in\Lambda_D$ does not satisfy the modulo system defining $\Lambda_{D'}$, so we have $\Lambda_D\not\subset\Lambda_{D'}$. We conclude that there is no set relation between $\Lambda_D$ and $\Lambda_{D'}$.
\item[{\textit{iv)} $\Lambda_{D'}\subsetneqq\Lambda_{CF}$:}] The lattice $\Lambda_{D'}$ has generators $(1,-1,0,0)$, $(1,0,-1,0)$, $(2,0,0,2)$, and $(4,0,0,0)$, all of which is in $\Lambda_{CF}$, so $\Lambda_{D'}\subseteq\Lambda_{CF}$. However, $(1,1,0,0)\in\Lambda_{CF}$ does not satisfy the modulo system defining $\Lambda_{D'}$, so we have $\Lambda_{D'}\subsetneqq\Lambda_{CF}$.
\end{description}
Hence, we conclude that the lattices $\Lambda_D$ and $\Lambda_{D'}$ are not the same and are strictly contained in $\Lambda_{CF}$. {We note, though, that in general $\Lambda_{CF}$ contains $\Gamma_{CF}$ and $\Lambda_D$, but not necessarily $\Lambda_{D'}$.}

\remove{Since the only element in $\psi(C_0)+2\psi(C_{1})$ that is not in $\Lambda_D$ is $(0,1,1,0)$ and $(0,1,1,0)=(1,1,0,0)+(-1,0,1,0)$ where $(1,1,0,0)\in\Lambda_D$ and $(-1,0,1,0)\in\Lambda_{D'}$, we conclude that $\Lambda_{CF}=\Lambda_D+\Lambda_{D'}$ (see also item 4 in Remark \ref{Remark}).}
\end{example}

%******************************************************************************%
%
%
%
%******************************************************************************%

\section{Connections between Construction D and Construction by Code Formula}

In this section, we give an explicit description of the lattices $\Lambda_{CF}$ constructed using Construction by Code Formula and further relate this construction to Construction D using Schur product of codes. Denote by $\ast$ componentwise multiplication (known also as Schur product or Hadamard product). That is, for $\mathbf{x}=(x_1,\ldots,x_n),\mathbf{y}=(y_1,\ldots,y_n)\in\mathbb{F}_2^n$, we have
\[\mathbf{x}\ast\mathbf{y}:=(x_1y_1,\ldots,x_ny_n)\in\mathbb{F}_2^n.\]
It is not hard to see that
\remove{$\mbox{supp}(\mathbf{x}\ast\mathbf{y})=\mbox{supp}(\mathbf{x})\cap\mbox{supp}(\mathbf{y})$ and \\
Here, note that the addition $\mathbf{x}+\mathbf{y}$ is over $\mathbb{F}_2$.}
\begin{equation}\label{EquationAst}
\psi(\mathbf{x})+\psi(\mathbf{y})=\psi(\mathbf{x}+\mathbf{y})+2\psi(\mathbf{x}\ast\mathbf{y})
\end{equation}
where additions are taken over the respective spaces.

We say that a family of nested binary linear codes $C_0\subseteq{}C_1\subseteq\ldots\subseteq{}C_{a-1}\subseteq{}C_a=\mathbb{F}_2^n$ is closed under Schur product if and only if the Schur product of any two codewords of $C_i$ is contained in $C_{i+1}$ for all $i$. In other words, if $\mathbf{c}_1,\mathbf{c}_2\in{}C_i$, then $\mathbf{c}_1\ast\mathbf{c}_2\in{}C_{i+1}$ for all $i=0,\ldots,a-1$. We now look at lattice closure of the code formula $\Gamma_{CF}=\psi(C_0)+2\psi(C_1)+\ldots+2^{a-1}\psi(C_{a-1})+2^a\mathbb{Z}^n$.

\begin{proposition}
Let $C_0\subseteq{}C_1\subseteq\ldots\subseteq{}C_{a-1}\subseteq{}C_a=\mathbb{F}_2^n$ be a family of nested binary linear codes. The smallest lattice $\Lambda_{CF}$ containing $\Gamma_{CF}$ consists of all vectors of the form
\begin{equation}\label{LatticeDbar}
\sum_{i=0}^{a-1}{{2^i}\sum_{\mathbf{c}_j\in{}C_i}{\alpha_j^{(i)}\psi(\mathbf{c}_j)}}+2^a\mathbf{l}
\end{equation}
where $\alpha_j^{(i)}\in\{0,1\}$ and $\mathbf{l}\in\mathbb{Z}^n$.
\end{proposition}

\begin{proof}
Let
\[\Lambda=\left\{\sum_{i=0}^{a-1}{{2^i}\sum_{\mathbf{c}_j\in{}C_i}{\alpha_j^{(i)}\psi(\mathbf{c}_j)}}+2^a\mathbf{l}\vbar\alpha_j^{(i)}\in\{0,1\}\mbox{ and }\mathbf{l}\in\mathbb{Z}^n\right\}.\]
It is not hard to see that $\Lambda$ is a lattice and $\Lambda$ contains $\Gamma_{CF}$. Let $\Lambda'$ be a lattice such that $\Gamma_{CF}\subseteq\Lambda'$, and let $\mathbf{v}\in\Lambda$. One may express $\mathbf{v}$ as
\[\sum_{i=0}^{a-1}{{2^i}\sum_{\mathbf{c}_j\in{}C_i}{\alpha_j^{(i)}\psi(\mathbf{c}_j)}}+2^a\mathbf{l}\]
where $\alpha_j^{(i)}\in\{0,1\}$ and $\mathbf{l}\in\mathbb{Z}^n$. Now, since $2^a\mathbf{l}\in\Gamma_{CF}$ and $2^i\alpha_j^{(i)}\psi(\mathbf{c}_j)\in\Gamma_{CF}$ for all $i\in\{0,1,\ldots,a-1\}$ and $\mathbf{c}_j\in{}C_i$, we must have $\mathbf{v}\in\Lambda'$. We conclude that $\Lambda$ is the smallest lattice that contains $\Gamma_{CF}$ and hence $\Lambda=\Lambda_{CF}$. \qed
\end{proof}

\begin{remark}
The expression \eqref{LatticeDbar} of an element in $\Lambda_{CF}$ is often not unique.
\end{remark}

Note that Construction D depends on the chosen basis $\mathbf{b}_1,\mathbf{b}_2,\ldots,\mathbf{b}_n$ of $\mathbb{F}_2^n$ where $\mathbf{b}_1,\ldots,\mathbf{b}_{k_i}$ span $C_i$. Hence, to a choice of codes corresponds (possibly) several lattices. The following corollary states that the sum of all such lattices yields the lattice $\Lambda_{CF}$ from Construction by Code Formula using the same nested codes.

\begin{corollary}\label{CorollaryDinDbar}
Let $\mathfrak{L}$ be the set of all lattices constructible from a family of nested binary linear codes $C_0\subseteq{}C_1\subseteq\ldots\subseteq{}C_{a-1}\subseteq{}C_a=\mathbb{F}_2^n$ using Construction D. Then,
\[\bigoplus_{\Lambda_{D}\in\mathfrak{L}}{\Lambda_{D}}=\Lambda_{CF}.\]
\end{corollary}

In the following theorem, we give a sufficient and necessary condition for
\[\Gamma_{CF}=\psi(C_0)+2\psi(C_1)+\ldots+2^{a-1}\psi(C_{a-1})+2^a\mathbb{Z}^n\]
from Construction by Code Formula to be a lattice. In addition, when the condition is met, the resulting lattice is the same as the lattice from Construction D.

\begin{theorem}\label{Theorem}
Given a family of nested binary linear codes $C_0\subseteq{}C_1\subseteq\ldots\subseteq{}C_{a-1}\subseteq{}C_a=\mathbb{F}_2^n$, the following statements are equivalent.
\begin{enumerate}
\item $\Gamma_{CF}$ is a lattice.
\item $\Gamma_{CF}=\Lambda_{CF}$.
\item $C_0\subseteq{}C_1\subseteq\ldots\subseteq{}C_{a-1}\subseteq{}C_a=\mathbb{F}_2^n$ is closed under Schur product.
\item $\Gamma_{CF}=\Lambda_D$.
\end{enumerate}
\end{theorem}

\begin{proof}
It is easy to see that $\mathit{4}\Rightarrow\mathit{1}\Rightarrow\mathit{2}$. We have left to show that $\mathit{2}\Rightarrow\mathit{3}$ and $\mathit{3}\Rightarrow\mathit{4}$. \\

Suppose that $C_0\subseteq{}C_1\subseteq\ldots\subseteq{}C_{a-1}\subseteq{}C_a=\mathbb{F}_2^n$ is not closed under Schur product. That is, there exist $\mathbf{c}_1,\mathbf{c}_2\in{}C_i$ such that $\mathbf{c}_1\ast\mathbf{c}_2\notin{}C_{i+1}$ for some $i$. Therefore, $2^{i+1}\psi(\mathbf{c}_1\ast\mathbf{c}_2)\notin\Gamma_{CF}$. On the other hand, it follows from \eqref{EquationAst} that
\[2^{i+1}\psi(\mathbf{c}_1\ast\mathbf{c}_2)=2^i\psi(\mathbf{c}_1)+2^i\psi(\mathbf{c}_2)-2^i\psi(\mathbf{c}_1+\mathbf{c}_2)\in\Lambda_{CF}.\]
Thus, we have $\Gamma_{CF}\neq\Lambda_{CF}$. \\

Now, we will prove $\mathit{3}\Rightarrow\mathit{4}$ by induction. The case $a=1$ is trivial since both $\Gamma_{CF}$ and $\Lambda_D$ coincide with the lattice constructed from $C_0$ using Construction A. Let $C_0\subseteq{}C_1\subseteq\ldots\subseteq{}C_{a-1}\subseteq{}C_a=\mathbb{F}_2^n$ be a family of nested binary linear codes that is closed under Schur product, and let $\mathbf{b}_1,\mathbf{b}_2,\ldots,\mathbf{b}_n$ be a basis of $\mathbb{F}_2^n$ such that $\mathbf{b}_1,\ldots,\mathbf{b}_{k_i}$ span $C_i$. Applying induction hypothesis to $C_1\subseteq\ldots\subseteq{}C_{a-1}\subseteq{}C_a=\mathbb{F}_2^n$ yields
\begin{equation}\label{Lattice}
\Gamma'_{CF}=\Lambda'_D
\end{equation}
where
\[\Gamma'_{CF}=2\psi(C_{1})+\ldots+2^{a-1}\psi(C_{a-1})+2^a\mathbb{Z}^n\]
and
\[\Lambda'_D=\left\{\sum_{i=1}^{a-1}{{2^i}\sum_{j=1}^{k_i}{\alpha_j^{(i)}\psi(\mathbf{b}_j)}}+2^a\mathbf{l}\vbar\alpha_j^{(i)}\in\{0,1\}\mbox{ and }\mathbf{l}\in\mathbb{Z}^n\right\}.\]
To avoid confusion, we denote by $\Lambda$ the lattice given in (\ref{Lattice}). We now wish to show that
\[\Gamma_{CF}=\psi(C_0)+\Lambda=\{\psi(\mathbf{c})+\mathbf{a}\mid\mathbf{c}\in{}C_0\mbox{ and }\mathbf{a}\in\Lambda\}\]
is equal to
\[\Lambda_D=\left\{\sum_{j=1}^{k_0}{\alpha_j^{(0)}\psi(\mathbf{b}_j)}+\mathbf{a}\vbar\alpha_j^{(0)}\in\{0,1\}\mbox{ and }\mathbf{a}\in\Lambda\right\}.\]

To do so, we will prove by induction that if $\mathbf{c}\in{}C_0$ is a binary sum of $\mathbf{b}_{j_1},\ldots,\mathbf{b}_{j_s}$, $1\leq{}j_1,\ldots,j_s\leq{}k_0$, then
\[\psi(\mathbf{b}_{j_1})+\ldots+\psi(\mathbf{b}_{j_s})=\psi(\mathbf{c})+\mathbf{a}\]
for some $\mathbf{a}\in\Lambda$. The case $s=1$ is trivial since one may pick $\mathbf{a}=\mathbf{0}\in\Lambda$. Let $\mathbf{c}\in{}C_0$ be a binary sum of $\mathbf{b}_{j_1},\ldots,\mathbf{b}_{j_s}$, $1\leq{}j_1,\ldots,j_s\leq{}k_0$. By induction hypothesis, there exists $\mathbf{a'}\in\Lambda$ such that
\[\psi(\mathbf{b}_{j_1})+\ldots+\psi(\mathbf{b}_{j_{s-1}})=\psi(\mathbf{c}')+\mathbf{a}'\]
where $\mathbf{c}'\in{}C_0$ is a binary sum of $\mathbf{b}_{j_1},\ldots,\mathbf{b}_{j_{s-1}}$. Now,
\begin{align*}
\psi(\mathbf{b}_{j_1})+\ldots+\psi(\mathbf{b}_{j_{s-1}})+\psi(\mathbf{b}_{j_s}) & =\psi(\mathbf{c}')+\psi(\mathbf{b}_{j_s})+\mathbf{a}' \\
& =\psi(\mathbf{c}'+\mathbf{b}_{j_s})+2\psi(\mathbf{c}'\ast\mathbf{b}_{j_s})+\mathbf{a}' \\
& =\psi(\mathbf{c})+2\psi(\mathbf{c}'\ast\mathbf{b}_{j_s})+\mathbf{a}'.
\end{align*}
Since $C_0\subseteq{}C_1$ is closed under Schur product, $\mathbf{c}'\ast\mathbf{b}_{j_s}\in{}C_1$, and so $2\psi(\mathbf{c}'\ast\mathbf{b}_{j_s})\in\Lambda$. Letting $\mathbf{a}=2\psi(\mathbf{c}'\ast\mathbf{b}_{j_s})+\mathbf{a}'\in\Lambda$, we obtain
\[\psi(\mathbf{b}_{j_1})+\ldots+\psi(\mathbf{b}_{j_s})=\psi(\mathbf{c})+\mathbf{a}\]
as desired. We can now conclude that $\Gamma_{CF}=\Lambda_D$, and this finishes the proof of the theorem. \qed
\end{proof}

The above theorem explains why Construction D and Construction by Code Formula have both been successful in constructing and encoding Barnes-Wall lattices. This is due to the fact that a family of Reed-Muller codes is closed under Schur product. For the same reason, Construction D produces a unique lattice despite many choices for the basis of Reed-Muller codes. This is summarized in the corollaries below.

\begin{corollary}\label{CorollaryDsameDbar}
If a family of nested binary linear codes $C_0\subseteq{}C_1\subseteq\ldots\subseteq{}C_{a-1}\subseteq{}C_a=\mathbb{F}_2^n$ is closed under Schur product, then Construction D and the code formula yield the same lattice.
\end{corollary}

\begin{corollary}
If a family of nested binary linear codes $C_0\subseteq{}C_1\subseteq\ldots\subseteq{}C_{a-1}\subseteq{}C_a=\mathbb{F}_2^n$ is closed under Schur product, then the resulting lattices from Construction D are the same irrespective of the chosen basis $\mathbf{b}_1,\mathbf{b}_2,\ldots,\mathbf{b}_n$ of $\mathbb{F}_2^n$ such that $\mathbf{b}_1,\ldots,\mathbf{b}_{k_i}$ span $C_i$.
\end{corollary}

We also would like to note that the property of being closed under Schur product generalizes Forney's concept of ``carries'' \cite[page 1133]{F1}. It was shown that if a lattice is decomposable as $\Lambda=\psi(C_0)+2\psi(C_1)+4\mathbb{Z}^n$, then the carries of the sum (i.e., Schur product) of any codewords of $C_0$ is in $C_1$.

If a chain of codes $C_0\subseteq{}C_1\subseteq\ldots\subseteq{}C_{a-1}\subseteq{}C_a=\mathbb{F}_2^n$ is not closed under Schur product, then we know from Theorem \ref{Theorem} that $\Gamma_{CF}=\psi(C_0)+2\psi(C_1)+\ldots+2^{a-1}\psi(C_{a-1})+2^a\mathbb{Z}^n$ is not a lattice. However, one may easily construct another chain of codes $C'_0\subseteq{}C'_1\subseteq\ldots\subseteq{}C'_{a-1}\subseteq{}C'_a=\mathbb{F}_2^n$ which is closed under Schur product and $C_i\subseteq{}C'_i$ for all $i$. It follows that $\Gamma'_{CF}=\psi(C'_0)+2\psi(C'_1)+\ldots+2^{a-1}\psi(C'_{a-1})+2^a\mathbb{Z}^n$ contains $\Gamma_{CF}$ and is itself a lattice. Since $\Lambda_{CF}$ is defined as the smallest lattice that contains $\Gamma_{CF}$, we have $\Lambda_{CF}\subseteq\Gamma'_{CF}$. The next example will demonstrate that $\Lambda_{CF}$ can be made a strict subset of $\Gamma'_{CF}$. In other words, $\Gamma'_{CF}$ may not be the smallest lattice that contains $\Gamma_{CF}$ despite the minimality of $C'_i$.

\begin{example}
Let $C$ be the shortened first order Reed-Muller code of length 15; that is, $C$ has a parity check matrix
\[\left(
\begin{array}{cccccccccccccccccc}
0 & 0 & 0 & & 0 & 0 & 0 & 0 & & 1 & 1 & 1 & 1 & & 1 & 1 & 1 & 1 \\
0 & 0 & 0 & & 1 & 1 & 1 & 1 & & 0 & 0 & 0 & 0 & & 1 & 1 & 1 & 1 \\
0 & 1 & 1 & & 0 & 0 & 1 & 1 & & 0 & 0 & 1 & 1 & & 0 & 0 & 1 & 1 \\
1 & 0 & 1 & & 0 & 1 & 0 & 1 & & 0 & 1 & 0 & 1 & & 0 & 1 & 0 & 1
\end{array}\right)\]
and is known also as the simplex code, the dual of the Hamming code of length 15. The chain of codes $C\subseteq{}C\subseteq{}C\subseteq\mathbb{F}_2^{15}$ is not closed under Schur product. So, $\Gamma_{CF}=\psi(C)+2\psi(C)+4\psi(C)+8\mathbb{Z}^{15}$ is not a lattice and is a strict subset of the lattice $\Lambda_{CF}$. We note here that the sum of the entries of an element in this lattice is divisible by 8.

Now, we denote by $C_1$ the smallest code that contains the Schur product of elements in $C$ and $C_2$ the smallest code that contains the Schur product of elements in $C_2$. It follows that $C_1$ and $C_2$ are the shortened second and third order Reed-Muller code and have a parity check matrix
\[\left(
\begin{array}{cccccccccccccccccc}
0 & 0 & 0 & & 0 & 0 & 0 & 0 & & 1 & 1 & 1 & 1 & & 1 & 1 & 1 & 1 \\
0 & 0 & 0 & & 1 & 1 & 1 & 1 & & 0 & 0 & 0 & 0 & & 1 & 1 & 1 & 1 \\
0 & 1 & 1 & & 0 & 0 & 1 & 1 & & 0 & 0 & 1 & 1 & & 0 & 0 & 1 & 1 \\
1 & 0 & 1 & & 0 & 1 & 0 & 1 & & 0 & 1 & 0 & 1 & & 0 & 1 & 0 & 1 \\[3pt]
0 & 0 & 0 & & 0 & 0 & 0 & 0 & & 0 & 0 & 0 & 0 & & 1 & 1 & 1 & 1 \\
0 & 0 & 0 & & 0 & 0 & 0 & 0 & & 0 & 0 & 1 & 1 & & 0 & 0 & 1 & 1 \\
0 & 0 & 0 & & 0 & 0 & 0 & 0 & & 0 & 1 & 0 & 1 & & 0 & 1 & 0 & 1 \\
0 & 0 & 0 & & 0 & 0 & 1 & 1 & & 0 & 0 & 0 & 0 & & 0 & 0 & 1 & 1 \\
0 & 0 & 0 & & 0 & 1 & 0 & 1 & & 0 & 0 & 0 & 0 & & 0 & 1 & 0 & 1 \\
0 & 0 & 1 & & 0 & 0 & 0 & 1 & & 0 & 0 & 0 & 1 & & 0 & 0 & 0 & 1
\end{array}\right)\mbox{ and }\left(
\begin{array}{cccccccccccccccccc}
0 & 0 & 0 & & 0 & 0 & 0 & 0 & & 1 & 1 & 1 & 1 & & 1 & 1 & 1 & 1 \\
0 & 0 & 0 & & 1 & 1 & 1 & 1 & & 0 & 0 & 0 & 0 & & 1 & 1 & 1 & 1 \\
0 & 1 & 1 & & 0 & 0 & 1 & 1 & & 0 & 0 & 1 & 1 & & 0 & 0 & 1 & 1 \\
1 & 0 & 1 & & 0 & 1 & 0 & 1 & & 0 & 1 & 0 & 1 & & 0 & 1 & 0 & 1 \\[3pt]
0 & 0 & 0 & & 0 & 0 & 0 & 0 & & 0 & 0 & 0 & 0 & & 1 & 1 & 1 & 1 \\
0 & 0 & 0 & & 0 & 0 & 0 & 0 & & 0 & 0 & 1 & 1 & & 0 & 0 & 1 & 1 \\
0 & 0 & 0 & & 0 & 0 & 0 & 0 & & 0 & 1 & 0 & 1 & & 0 & 1 & 0 & 1 \\
0 & 0 & 0 & & 0 & 0 & 1 & 1 & & 0 & 0 & 0 & 0 & & 0 & 0 & 1 & 1 \\
0 & 0 & 0 & & 0 & 1 & 0 & 1 & & 0 & 0 & 0 & 0 & & 0 & 1 & 0 & 1 \\
0 & 0 & 1 & & 0 & 0 & 0 & 1 & & 0 & 0 & 0 & 1 & & 0 & 0 & 0 & 1 \\[3pt]
0 & 0 & 0 & & 0 & 0 & 0 & 0 & & 0 & 0 & 0 & 0 & & 0 & 0 & 1 & 1 \\
0 & 0 & 0 & & 0 & 0 & 0 & 0 & & 0 & 0 & 0 & 0 & & 0 & 1 & 0 & 1 \\
0 & 0 & 0 & & 0 & 0 & 0 & 0 & & 0 & 0 & 0 & 1 & & 0 & 0 & 0 & 1 \\
0 & 0 & 0 & & 0 & 0 & 0 & 1 & & 0 & 0 & 0 & 0 & & 0 & 0 & 0 & 1
\end{array}\right)\]
respectively. We now have $C\subseteq{}C_1\subseteq{}C_2\subseteq\mathbb{F}_2^{15}$ closed under Schur product, and it follows that $\Gamma'_{CF}=\psi(C)+2\psi(C_1)+4\psi(C_2)+8\mathbb{Z}^{15}$ is a lattice. Nonetheless, $(0,0,2,0,0,0,2,0,0,0,2,2,2,2,0)\in2\psi(C_1)$ is not in $\Lambda_{CF}$ since the sum of its entries is 12. We conclude that $\Lambda_{CF}\subsetneqq\Gamma'_{CF}$.
\end{example}

\section{Construction A$'$}

In this section, we will consider both real and complex lattices, where {a complex lattice over $\mathbb{Z}[i]$ is a discrete subgroup of $\mathbb{C}^n$}. Denote by $\mathcal{R}$ either $\mathbb{Z}$ or $\mathbb{Z}[i]$, and let $v=2$ if $\mathcal{R}=\mathbb{Z}$ and $v=1+i$ if $\mathcal{R}=\mathbb{Z}[i]$. In other words, the constant $v$ takes on different values depending on the choice of $\mathcal{R}$. Again, let $\psi$ be the natural embedding of $\mathbb{F}_2^n$ into $\mathcal{R}^n$. We give a broader definition for Construction by Code Formula to include complex lattices below. The complex construction corresponds to the complex code formula given in \cite{F1}

\begin{definition}[Construction by Code Formula]
A lattice $\Lambda_{CF}$ over $\mathcal{R}$ is obtained from Construction by Code Formula using a family of nested binary linear codes $C_0\subseteq{}C_1\subseteq\ldots\subseteq{}C_{a-1}\subseteq{}C_a=\mathbb{F}_2^n$ if $\Lambda_{CF}$ is the smallest lattice that contains
\[\Gamma_{CF}=\psi(C_0)+v\psi(C_{1})+\ldots+v^{a-1}\psi(C_{a-1})+v^a\mathcal{R}^n.\]
\end{definition}

Define the polynomial quotient ring $\mathcal{U}_a:=\mathbb{F}_2[u]/u^a$ where $u$ is a variable. A linear code over $\mathcal{U}_a$ is a submodule of $\mathcal{U}_a^n$. The code $C$ corresponding to a generator matrix $G\in\mathcal{U}_a^{k\times{}n}$ is given by
\[C=\{\mathbf{u}G\mid\mathbf{u}\in\mathcal{U}_a^k\}\]
where matrix multiplication is over the ring $\mathcal{U}_a$ and $k$ is the rank of the code. One may embed $\mathcal{U}_a$ into $\mathcal{R}$ via the mapping $\Phi:\mathcal{U}_a\rightarrow\mathcal{R}$ given by
\[\Phi\left(\sum_{j=0}^{a-1}{b_ju^j}\right)=\sum_{j=0}^{a-1}{\psi(b_j)v^j}.\]
We will also use $\Phi$ as a bit-wise embedding from $\mathcal{U}_a^n$ into $\mathcal{R}^n$. The following construction is due to \cite{HVB_ISIT,HVB}.

\begin{definition}[Construction A$'$]
A lattice $\Lambda_{A'}$ over $\mathcal{R}$ is obtained from Construction A$'$ using a linear code $C$ over $\mathcal{U}_a$ if $\Lambda_{A'}$ is the smallest lattice that contains
\[\Gamma_{A'}=\Phi(C)+v^a\mathcal{R}^n.\]
\end{definition}

The next proposition shows that there exists a correspondence from a chain of binary linear codes $C_0\subseteq{}C_1\subseteq\ldots\subseteq{}C_{a-1}\subseteq{}C_a=\mathbb{F}_2^n$ to a linear code over $\mathcal{U}_a$ such that $\Gamma_{CF}$ and $\Gamma_{A'}$ from Construction by Code Formula and Construction A$'${, respectively,} coincide. This will prove that $\Gamma_{A'}$ is generally not a lattice, and any lattice constructible using Construction by Code Formula is also constructible using Construction A$'$.

\begin{proposition}\label{PropDbarA'}
Let $C_0\subseteq{}C_1\subseteq\ldots\subseteq{}C_{a-1}\subseteq{}C_a=\mathbb{F}_2^n$ be a family of nested binary linear codes and
\[\Gamma_{CF}=\psi(C_0)+v\psi(C_{1})+\ldots+v^{a-1}\psi(C_{a-1})+v^a\mathcal{R}^n.\]
There exists a linear code $C$ over $\mathcal{U}_a$ such that $\Gamma_{CF}=\Gamma_{A'}$ where
\[\Gamma_{A'}=\Phi(C)+v^a\mathcal{R}^n.\]
\end{proposition}

\begin{proof}
Let $k_i=\dim(C_i)$ and $k=k_{a-1}=\dim(C_{a-1})$. Let $G\in\mathbb{F}_2^{k\times{}n}$ be a generator matrix for $C_{a-1}$ such that the first $k_i$ rows of $G$ generate $C_i$. That is, one may write $G$ as
\[G=\left[
\begin{array}{c}
G_0 \\
G_{1} \\
\vdots \\
G_{a-1}
\end{array}
\right]\in\mathbb{F}_2^{k\times{}n}\]
where the $k_i\times{}n$ matrix $\left[
\begin{array}{c}
G_0 \\
G_1 \\
\vdots \\
G_i
\end{array}
\right]$ is a generator matrix for $C_i$. \\

Let $C$ be the code over $\mathcal{U}_a$ generated by a generator matrix
\[\widetilde{G}=\left[
\begin{array}{c}
G_0 \\
uG_1 \\
\vdots \\
u^{a-1}G_{a-1}
\end{array}
\right]\in\mathcal{U}_a^{k\times{}n},\]
and let
\[\Gamma_{A'}=\Phi(C)+v^a\mathcal{R}^n.\]
We will prove that
$\Gamma_{CF}=\Gamma_{A'}$. \\

Fix an element $\mathbf{x}\in\Gamma_{CF}$. By construction, one may express $\mathbf{x}$ as
\[\mathbf{x}=\psi(\mathbf{c}_0)+v\psi(\mathbf{c}_1)+\ldots+v^{a-1}\psi(\mathbf{c}_{a-1})+v^{a}\mathbf{l}\]
where $\mathbf{c}_i\in{}C_i$ and $\mathbf{l}\in\mathcal{R}^n$. Let $\mathbf{d}_i=(d_{i,1},\ldots,d_{i,k_i})\in\mathbb{Z}_2^{1\times{}k_i}$ be a vector such that
\[\mathbf{d}_i\left[
\begin{array}{c}
G_0 \\
G_1 \\
\vdots \\
G_i
\end{array}
\right]=(d_{i,1},\ldots,d_{i,k_i})\left[
\begin{array}{c}
G_0 \\
G_1 \\
\vdots \\
G_i
\end{array}
\right]=\mathbf{c}_i.\]

Now, one may multiply the entries of $\mathbf{d}_i$ by powers of $u$ and append 0 as necessary to obtain $\widetilde{\mathbf{d}}_i=(d_{i,1}u^i,\ldots,d_{i,k_i},0,\ldots,0)\in\mathcal{U}_a^{1\times{}k}$ such that
\[\widetilde{\mathbf{d}}_i\widetilde{G}=(d_{i,1}u^i,\ldots,d_{i,k_i},0,\ldots,0)\left[
\begin{array}{c}
G_0 \\
uG_1 \\
\vdots \\
u^iG_i \\
\vdots \\
u^{a-1}G_{a-1}
\end{array}
\right]=\mathbf{c}_iu^i.\]
Thus, we have
\[(\widetilde{\mathbf{d}}_0+\widetilde{\mathbf{d}}_1+\ldots+\widetilde{\mathbf{d}}_{a-1})\widetilde{G}=\mathbf{c}_0+\mathbf{c}_1u+\ldots+\mathbf{c}_{a-1}u^{a-1}.\]
It follows that
\[\Phi\left((\widetilde{\mathbf{d}}_0+\widetilde{\mathbf{d}}_1+\ldots+\widetilde{\mathbf{d}}_{a-1})\widetilde{G}\right)=\psi(\mathbf{c}_0)+\psi(\mathbf{c}_1)v+\ldots+\psi(\mathbf{c}_{a-1})v^{a-1},\]
and so $\mathbf{x}\in\Gamma_{A'}$. We may now conclude that $\Gamma_{CF}\subseteq\Gamma_{A'}$. \\

On the other hand, for any $\widetilde{\mathbf{d}}\in\mathcal{U}_a^{1\times{}k}$, the coefficient of $u^i$ in
\[\widetilde{\mathbf{d}}\widetilde{G}=\widetilde{\mathbf{d}}\left[
\begin{array}{c}
G_0 \\
uG_1 \\
\vdots \\
u^iG_i \\
\vdots \\
u^{a-1}G_{a-1}
\end{array}
\right]\]
must be a linear combination of the rows of $\left[
\begin{array}{c}
G_0 \\
G_1 \\
\vdots \\
G_{i}
\end{array}
\right]$. Therefore,
\[\widetilde{\mathbf{d}}\widetilde{G}=\mathbf{c}_0+\mathbf{c}_1u+\ldots+\mathbf{c}_iu^i+\ldots+\mathbf{c}_{a-1}u^{a-1}\]
for some $\mathbf{c}_0\in{}C_0,\ldots,\mathbf{c}_{a-1}\in{}C_{a-1}$. It follows that $\Gamma_{A'}\subseteq\Gamma_{CF}$, and this finishes the proof of the proposition. \qed
\end{proof}

\begin{corollary}\label{Corollary}
If a code $C$ over $\mathcal{U}_a$ can be expressed as $C_0+uC_1+\ldots+u^{a-1}C_{a-1}$ where $C_0\subseteq{}C_1\subseteq\ldots\subseteq{}C_{a-1}$ is closed under Schur product, then $\Gamma_{A'}=\Phi(C)+2^a\mathbb{Z}^n$ from Construction A$'$ for the reals is a lattice.
\end{corollary}

Note that the previous corollary provides us a glimpse of the necessary condition for Construction A$'$ to produce a lattice. To describe the precise condition, we now focus on Construction A$'$ for the reals and generalize Schur product to polynomials and elements in $\mathcal{U}_a^n$. The details are given as follows. \\

Let $x=x_0+x_1u+\ldots+x_{a-1}u^{a-1}$ and $y=y_0+y_1u+\ldots+y_{a-1}u^{a-1}$ be elements in $\mathcal{U}_a=\mathbb{F}_2[u]/u^a$. The entrywise multiplication of $x$ and $y$ (known also as Schur product or Hadamard product of polynomials) is given by
\[x\ast{}y=x_0y_0+x_1y_1u+x_{a-1}y_{a-1}u^{a-1}\in\mathcal{U}_a.\]
Now, for $\mathbf{w}=(w_1,\ldots,w_n),\mathbf{z}=(z_1,\ldots,z_n)\in\mathcal{U}_a^n$, define
\[\mathbf{w}\ast\mathbf{z}=(w_1\ast{}z_1,\ldots,w_n\ast{}z_n).\]

Alternatively, one may write $\mathbf{w},\mathbf{z}\in\mathcal{U}_a^n$ as $\mathbf{w}=\mathbf{w}_0+\mathbf{w}_1u+\ldots+\mathbf{w}_{a-1}u^{a-1}$ and $\mathbf{z}=\mathbf{z}_0+\mathbf{z}_1u+\ldots+\mathbf{z}_{a-1}u^{a-1}$ where $\mathbf{w}_i,\mathbf{z}_i\in\mathbb{F}_2^n$ for $i=0,\ldots,a-1$. It follows that
\[\mathbf{w}\ast\mathbf{z}=(\mathbf{w}_0\ast\mathbf{z}_0)+(\mathbf{w}_1\ast\mathbf{z}_1)u+\ldots+(\mathbf{w}_{a-1}\ast\mathbf{z}_{a-1})u^{a-1}.\]

We say that a code $C$ over $\mathcal{U}_a$ is closed under \textit{shifted} Schur product if and only if, for any codewords $\mathbf{c}_1$ and $\mathbf{c}_2$ of $C$, $(\mathbf{c}_1\ast\mathbf{c}_2)u$ is a codeword of $C$. We now state the necessary and sufficient condition for Construction A$'$ to give a lattice in the following theorem.

\begin{theorem}
Let $C$ be a code over $\mathcal{U}_a$. The set $\Gamma_{A'}=\Phi(C)+2^a\mathbb{Z}^n$ from Construction A$'$ is a lattice if and only if $C$ is closed under shifted Schur product.
\end{theorem}

\begin{proof}
We will first prove that if $\mathbf{c}_1$ and $\mathbf{c}_2$ are codewords of $C$, then
\begin{equation}\label{EquationSum}
\Phi(\mathbf{c}_1)+\Phi(\mathbf{c}_2)-\Phi(\mathbf{c}_1+\mathbf{c}_2)=\Phi((\mathbf{c}_1\ast\mathbf{c}_2)u)+2^a\mathbf{l}
\end{equation}
for some $\mathbf{l}\in\mathbb{Z}^n$. Write $\mathbf{c}_i$ as
\[\mathbf{c}_i=\mathbf{c}_{i,0}+\ldots+\mathbf{c}_{i,a-2}u^{a-2}+\mathbf{c}_{i,a-1}u^{a-1}\]
where $\mathbf{c}_{i,0},\ldots,\mathbf{c}_{i,a-2},\mathbf{c}_{i,a-1}\in\mathbb{F}_2^n$ and $i=1$ or 2; therefore,
\[\Phi(\mathbf{c}_i)=\psi(\mathbf{c}_{i,0})+\ldots+\psi(\mathbf{c}_{i,a-2})2^{a-2}+\psi(\mathbf{c}_{i,a-1})2^{a-1}.\]
Also,
\[\Phi(\mathbf{c}_1+\mathbf{c}_2)=\psi(\mathbf{c}_{1,0}+\mathbf{c}_{2,0})+\ldots+\psi(\mathbf{c}_{1,a-2}+\mathbf{c}_{2,a-2})2^{a-2}+\psi(\mathbf{c}_{1,a-1}+\mathbf{c}_{2,a-1})2^{a-1}.\]
Now, it follows from \eqref{EquationAst} that
\begin{align*}
\Phi(\mathbf{c}_1)+\Phi(\mathbf{c}_2)-\Phi(\mathbf{c}_1+\mathbf{c}_2) & =\psi(\mathbf{c}_{1,0}\ast\mathbf{c}_{2,0})2+\ldots+\psi(\mathbf{c}_{1,a-2}\ast\mathbf{c}_{2,a-2})2^{a-1} \\
& \quad+\psi(\mathbf{c}_{1,a-1}\ast\mathbf{c}_{2,a-1})2^a.
\end{align*}
Since $\Phi((\mathbf{c}_1\ast\mathbf{c}_2)u)=\psi(\mathbf{c}_{1,0}\ast\mathbf{c}_{2,0})2+\ldots+\psi(\mathbf{c}_{1,a-2}\ast\mathbf{c}_{2,a-2})2^{a-1}$, the desired result follows. \\

If $C$ is not closed under shifted Schur product, then there exist $\mathbf{c}_1,\mathbf{c}_2\in{}C$ such that $(\mathbf{c}_1\ast\mathbf{c}_2)u\notin{}C$. Since $\Phi(\mathbf{c}_1),\Phi(\mathbf{c}_2),\Phi(\mathbf{c}_1+\mathbf{c}_2)\in\Gamma_{A'}$ but $\Phi(\mathbf{c}_1)+\Phi(\mathbf{c}_2)-\Phi(\mathbf{c}_1+\mathbf{c}_2)$ does not reduce mod $2^a$ to an element in $\Phi(C)$, we conclude that $\Gamma_{A'}$ is not a lattice. \\

We are left to show that if $C$ is closed under Schur product then $\Gamma_{A'}=\Phi(C)+2^a\mathbb{Z}^n$ from Construction A$'$ is a lattice. It is not hard to see that $\Gamma_{A'}$ is discrete. Let $\mathbf{c}_1,\mathbf{c}_2\in{}C$. The smallest degree of $c_0+c_1u+\ldots+c_{a-1}u^{a-1}\in\mathcal{U}_a$ is the smallest $i$ such that $c_i\neq0$, with a convention that the smallest degree of $0$ is $a$. We say that $k$ is the smallest degree of $\mathbf{c}_2$ if the minimum of the smallest degree of entries of $\mathbf{c}_2$ is $k$. We will prove by backward induction on $k$ that $\Phi(\mathbf{c}_1)+\Phi(\mathbf{c}_2)\in\Gamma_{A'}$.

If the smallest degree of $\mathbf{c}_2$ is $a$ then $\mathbf{c}_2=\mathbf{0}$ and the result is obvious. If the smallest degree of $\mathbf{c}_2$ is $a-1$, then it follows from \eqref{EquationSum} that
\[\Phi(\mathbf{c}_1)+\Phi(\mathbf{c}_2)=\Phi(\mathbf{c}_1+\mathbf{c}_2)+\Phi((\mathbf{c}_1\ast\mathbf{c}_2)u)+2^a\mathbf{l}\]
for some $\mathbf{l}\in\mathbb{Z}^n$. Since $(\mathbf{c}_1\ast\mathbf{c}_2)u=\mathbf{0}$ in $\mathcal{U}_a$, we have that \[\Phi(\mathbf{c}_1)+\Phi(\mathbf{c}_2)=\Phi(\mathbf{c}_1+\mathbf{c}_2)+2^a\mathbf{l}\in\Gamma_{A'}.\]
Suppose now that the smallest degree of $\mathbf{c}_2$ is $0\leq{}k<a-1$. Again, we have
\[\Phi(\mathbf{c}_1)+\Phi(\mathbf{c}_2)=\Phi(\mathbf{c}_1+\mathbf{c}_2)+\Phi((\mathbf{c}_1\ast\mathbf{c}_2)u)+2^a\mathbf{l}\]
for some $\mathbf{l}\in\mathbb{Z}^n$. Since the smallest degree of $(\mathbf{c}_1\ast\mathbf{c}_2)u$ is $k+1$, we apply the induction hypothesis to $\mathbf{c}_1+\mathbf{c}_2$ and $(\mathbf{c}_1\ast\mathbf{c}_2)u$ and conclude that $\Phi(\mathbf{c}_1)+\Phi(\mathbf{c}_2)\in\Gamma_{A'}$. It readily follows that $\Gamma_{A'}$ is closed under addition.

Finally, for any $\mathbf{c}\in{}C$, we have
\[-\Phi(\mathbf{c})=(2^a-1)\Phi(\mathbf{c})-2^a\Phi(\mathbf{c})\in\Gamma_{A'}\]
since $(2^a-1)\Phi(\mathbf{c})\in\Gamma_{A'}$ and $\Phi(\mathbf{c})\in\mathbb{Z}^n$. \qed
\end{proof}

Proposition \ref{PropDbarA'}, Corollary \ref{Corollary}, and the above theorem relate in the following manners. Given a chain of binary codes $C_0\subseteq{}C_1\subseteq\ldots\subseteq{}C_{a-1}$,
\[C=C_0+uC_1+\ldots+u^{a-1}C_{a-1}\]
can be seen as a code over $\mathcal{U}_a$. Under such perspective, $\Phi(C)+2^a\mathbb{Z}^n$ from Construction A$'$ and the code formula $\psi(C_0)+2\psi(C_1)+\ldots+2^{a-1}\psi(C_{a-1})+2^a\mathbb{Z}^n$ coincide. If $C_0\subseteq{}C_1\subseteq\ldots\subseteq{}C_{a-1}$ is closed under Schur product, then $C$ is closed under shifted Schur product, and so $\Phi(C)+2^a\mathbb{Z}^n$ must be a lattice. Nonetheless, since the converse of the first statement is not true (i.e., a code over $\mathcal{U}_a$ is not necessarily a direct sum of binary codes), one may apply Construction A$'$ to such a code and obtain a lattice that is not constructible using Construction by Code Formula. The next example will demonstrate this presumption.

\begin{example}\label{ExampleL}
Consider a code $C$ of length 2 over $\mathcal{U}_3$ generated by a generator matrix
\[[
\begin{array}{cc}
1+u & 1+u+u^2
\end{array}].\]
Then, we have
\begin{align*}
C=\{ & (0,0),(u,u),(u^2,u^2),(u+u^2,u+u^2), \\
& (1,1+u^2),(1+u^2,1),(1+u,1+u+u^2),(1+u+u^2,1+u)\},
\end{align*}
and it is not hard to see that $C$ is closed under shifted Schur product. We apply Construction A$'$ over the reals and obtain the lattice
\[\{(0,0),(2,2),(4,4),(6,6),(1,5),(5,1),(3,7),(7,3)\}+8\mathbb{Z}^2\]
as shown in Figure 1.
\begin{figure}\label{FigureLattice}
\begin{center}
\begin{picture}(125,125)
\multiput(0,0)(15,0){9}{\line(0,1){120}}
\multiput(0,0)(0,15){9}{\line(1,0){120}}
\put(0,0){\circle*{5}}
\put(0,120){\circle*{5}}
\put(120,0){\circle*{5}}
\put(120,120){\circle*{5}}
\put(15,75){\circle*{5}}
\put(30,30){\circle*{5}}
\put(45,105){\circle*{5}}
\put(60,60){\circle*{5}}
\put(75,15){\circle*{5}}
\put(90,90){\circle*{5}}
\put(105,45){\circle*{5}}
\end{picture}
\caption{Lattice from Example \ref{ExampleL}}
\end{center}
\end{figure}
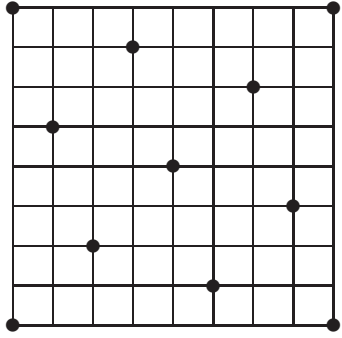
Note that this lattice is not constructible by Construction D, Construction D$'$, or Construction by Code Formula. {However, it was mentioned by a reviewer that this lattice can be obtained via Construction A using the 8-ary code $C=\langle(\bar{1},\bar{5})\rangle\subseteq\mathbb{Z}_8^2$.}
\end{example}

\section{Conclusion}

\begin{figure}\label{FigDiagram}
\begin{center}
\begin{picture}(315,150)
\put(10,40){\line(1,0){80}}
\put(10,40){\line(0,1){20}}
\put(10,60){\line(1,0){80}}
\put(90,40){\line(0,1){20}}
\put(20,47){Code Formula}
\put(10,30){- produces lattices iff codes are}
\put(16,20){closed under Schur product${}^1$}
\put(205,40){\line(1,0){80}}
\put(205,40){\line(0,1){20}}
\put(205,60){\line(1,0){80}}
\put(285,40){\line(0,1){20}}
\put(215,47){Construction A$'$}
\put(205,30){- uses codes over $\mathbb{F}_2[u]/u^a$}
\put(205,20){- produces lattices iff the}
\put(211,10){code is closed under}
\put(211,0){shifted Schur product}
\put(10,120){\line(1,0){80}}
\put(10,120){\line(0,1){20}}
\put(10,140){\line(1,0){80}}
\put(90,120){\line(0,1){20}}
\put(20,127){Construction D}
\put(10,110){- always produces lattices}
\put(205,120){\line(1,0){80}}
\put(205,120){\line(0,1){20}}
\put(205,140){\line(1,0){80}}
\put(285,120){\line(0,1){20}}
\put(215,127){Construction D$'$}
\put(205,110){- always produces lattices}
\put(0,50){\vector(1,0){10}}
\put(0,50){\line(0,1){80}}
\put(0,130){\vector(1,0){10}}
\put(4,75){equivalent iff ${}^1$ holds}
\put(90,50){\vector(1,0){115}}
\put(99,54){can be constructed using}
\put(90,130){\line(1,0){115}}
\put(147,130){\vector(0,-1){28}}
\put(90,60){\vector(3,2){30}}
\put(100,95){produce Barnes-Wall lattices}
\put(100,85){from Reed-Muller codes}
\end{picture}
\caption{Relationships between the four constructions of lattices}
\end{center}
\end{figure}
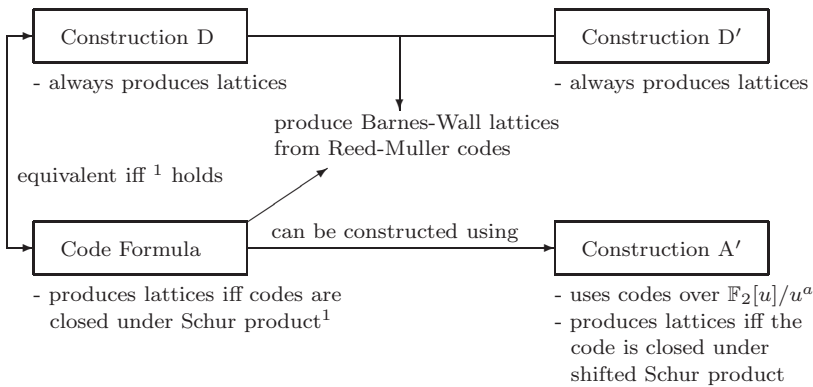

{Figure 2 summarizes the connections} between Construction D, Construction D$'$, and Construction by Code Formula of lattices from nested binary linear codes and Construction A$'$ of lattices from codes over the polynomial ring $\mathbb{F}_2[u]/u^a$. While Construction by Code Formula and Construction A$'$ are useful in the construction of Barnes-Wall lattices, we demonstrate that they may generate nonlattice packings. We give an exact condition for the two constructions to yield a lattice and further relate Construction by Code Formula to Construction D. Future work involves studying complex counterparts of Construction by Code Formula and Construction A$'$, along with possible properties and parameters of lattices from these constructions. {Another direction is to derive a relationship between code formula and Construction D$'$ in a similar fashion to code formula and Construction D. In addition, we see that binary codes can be combined to create a code over $\mathcal{U}_a$ where the two objects produce an identical lattice. It would be interesting to see if binary codes can be weaved together to create a code over $\mathbb{Z}_{2^a}$ in such a way that lattices constructed from such a code correspond to lattices from the multilevel constructions.}

\begin{acknowledgements}
The authors would like to thank Cong Ling and Jagadeesh Harshan for their helpful comments and suggestions.
\end{acknowledgements}



\end{document}